\documentclass[11pt,a4paper]{article}
\usepackage[T1]{fontenc}
\usepackage{lmodern}
\usepackage[margin=25mm]{geometry}
\usepackage{amsmath,amssymb,amsthm,bm}
\usepackage{booktabs,array,microtype}
\usepackage{hyperref}
\hypersetup{hidelinks,pdftitle={Exceptional Jordan spectra and fermion masses: finite selectors and joint quark tests},pdfauthor={Tejinder P. Singh}}
\newtheorem{proposition}{Proposition}

\theoremstyle{remark}

\newcommand{\R}{\mathbb R}
\newcommand{\C}{\mathbb C}
\newcommand{\Hh}{\mathbb H}
\newcommand{\Oo}{\mathbb O}
\newcommand{\id}{\boldsymbol 1}
\newcommand{\Tr}{\operatorname{Tr}}
\newcommand{\diag}{\operatorname{diag}}
\newcommand{\Hom}{\operatorname{Hom}}
\newcommand{\Herm}{\operatorname{Herm}}
\newcommand{\Sym}{\operatorname{Sym}}
\newcommand{\Cl}{\operatorname{Cl}}
\newcommand{\MS}{\overline{\mathrm{MS}}}
\newcommand{\Ree}{\operatorname{Re}}
\newcommand{\calB}{\mathcal B}
\newcommand{\UpResidual}{-9.47498}
\newcommand{\UpBound}{0.77501}
\newcommand{\UpResidualHigh}{-11.57543}

\newcommand{\LeptonResidual}{0.92459}
\newcommand{\LeptonBound}{0.01217}

\newcommand{\LeptonBoundHigh}{0.01744}
\newcommand{\PDGStrangeDown}{19.76596}
\newcommand{\AnsatzStrangeDown}{17.30220}
\newcommand{\PDGBoxLow}{19.32914}
\newcommand{\PDGBoxHigh}{20.21598}
\newcommand{\PDGStrangeDownFromRatios}{19.97823}
\newcommand{\PDGTauMuon}{4.10094}
\newcommand{\PDGMuonElectron}{14.37944}
\newcommand{\RunningTauMuon}{4.12146}
\newcommand{\RunningMuonElectron}{14.51428}

\newcommand{\DirectMuonElectron}{1.19458}
\newcommand{\DirectSortedLow}{1.98815}
\newcommand{\DirectSortedHigh}{1.19458}
\newcommand{\PowerUpFit}{0.44893}
\newcommand{\DeltaUpFit}{0.66537}
\newcommand{\DownAtUpFit}{2.05531}
\newcommand{\UEMZ}{1.592}
\newcommand{\DEMZ}{2.355}
\newcommand{\UEHigh}{1.032}
\newcommand{\DEHigh}{1.553}
\newcommand{\ElectronTwoGeV}{0.4951}
\newcommand{\UETwoGeV}{2.089}
\newcommand{\DETwoGeV}{3.081}
\newcommand{\OldDownRatio}{4.35494}
\newcommand{\OldDownError}{1.09947}
\newcommand{\OldBottomRatio}{7.24882}
\newcommand{\OldBottomError}{1.02764}
\newcommand{\pD}{1.0485}
\newcommand{\pB}{1.0444}
\newcommand{\pC}{0.9853}
\newcommand{\pT}{1.1174}
\newcommand{\pL}{0.9935}
\newcommand{\pE}{1.0110}

\title{Exceptional Jordan spectra and fermion masses:\\
finite selectors and joint quark tests}
\author{Tejinder P. Singh\\[2pt]
\small Tata Institute of Fundamental Research\\
\small Homi Bhabha Road, Mumbai 400005, India\\
\small Corresponding author: \href{mailto:tpsingh@tifr.res.in}{\texttt{tpsingh@tifr.res.in}}}
\date{20 November 2026}
\begin{document}
\maketitle

\begin{abstract}
We study an aligned direct-pairing construction of fermion masses from
exceptional Jordan spectra. A coassociative restriction on an equal-diagonal
Euclidean Albert element gives an equally spaced spectrum, while leaving its
dimensionless spacing free. Restriction of the auxiliary symmetric cube to
the monomial groups \(\Delta(3n^2)\) and \(\Delta(6n^2)\), for every
\(n\ge4\), classifies its invariant triplet selectors. Each triplet occurs
once, but covariant mass operators retain substantial freedom beyond a fixed
pairing. Our principal result is a complete joint quark test within this
aligned class. Keeping a common spectral spacing and the centre ratio
\(s_u/s_d=2/3\), we allow every real spacing, all invariant triplets and
independent positive hierarchy exponents in the two quark sectors.
Logarithmic mass-shape ratios eliminate both exponents and reduce the
down-sector constraint to seven inverse branches. Outward-rounded interval
bounds show that every branch fails the up-sector test in the published
marginal-input boxes at nine Standard Model benchmarks from the electroweak
scale to \(10^{16}\,\mathrm{GeV}\). The result also survives explicit padding
for printed-number rounding. These boxes carry no joint confidence level and
do not include model-dependent theoretical corrections. Earlier monomial
prescriptions and one-parameter modular models provide complementary
comparisons. The results constrain a specified mass law while leaving the
broader Clifford--Jordan programme open. Reproducibility code is supplied.
\end{abstract}
\noindent\textbf{Keywords:} Exceptional Jordan algebra; fermion masses;
finite flavour groups; covariant mass operators; interval arithmetic.

\section{The algebraic question}
\label{sec:intro}

Can an algebra determine fermion mass ratios as it can determine the eigenvalues
of a charge operator? Exceptional-Jordan proposals have associated charged-fermion
hierarchies with spectra in \(H_3(\Oo)\) and with products of their eigenvalues
\cite{Singh2022EPJP,Bhatt2022,Singh2022Physics}. The useful question is which
steps in that association are consequences of the stated algebraic data.
An eigenvalue calculation, an embedding of physical fermions, and a rule
turning spectral data into a Yukawa matrix are distinct mathematical operations.

This article gives a self-contained formulation and test of a particular
construction considered in earlier versions of this article and in
Ref.~\cite{TeliSingh2026}. Its algebraic starting point is an equally spaced
Jordan spectrum. Its effective mass map uses the multiplicative action on
\(\Sym^3\C^3\), selected weight states, and a positive pairing. This formulation
makes both its predictions and its freedom explicit. In particular, covariance,
positivity and homogeneity admit maps with different mass ratios, even when
null modes are preserved.

We replace arbitrary three-state choices by intertwiners of specified finite
groups. The monomial families also appear in modular flavour constructions
\cite{Chulia2026,EscobarLuhn}. Standard representation theory classifies their
restrictions of the symmetric cube for every \(n\ge4\). The application here
combines that classification with a complete sign-branch analysis and a
certified joint quark test, allowing separate positive power exponents in
the two sectors. It also distinguishes a unique embedding from a unique
background-dependent Yukawa operator. We make no claim to a new classification
of the finite groups themselves.

There is an elementary limitation to the charge analogy. In a fixed electric
charge sector, the three families carry
\begin{equation}
 Q_{\rm family}=q\,\id_3 .
 \label{eq:chargeidentity}
\end{equation}
A unitary change between flavour and mass bases leaves this operator unchanged.
The singular values of a mass matrix are likewise independent of the basis
used to describe the fermions. Measuring charge in a mass basis therefore does
not produce generation-dependent charges, and measuring masses through
interactions that identify flavour does not by itself generate a mass hierarchy.
A Clifford number operator can fix discrete charge eigenvalues after its
representation and normalization are specified; the corresponding mass problem
also requires a specified operator and its physical pairing.

The scope is consequently precise. The results below concern the stated
Jordan backgrounds, selectors and effective mass maps, with canonical kinetic
terms and a direct comparison to Standard Model (SM) running parameters.
The finite groups are additional input; their monomial frame is aligned with
the spectral frame, rather than derived from the Albert algebra.
They establish neither a complete Clifford--Jordan matter theory nor a general
obstruction to algebraic mass ratios. Additional algebraic constraints could
select a background and a coupling. Their explanatory value would depend on
being specified independently of the masses they are intended to predict.

\section{A real Jordan spectrum and its unfixed spacing}
\label{sec:jordan}

Let \(J=H_3(\Oo)\) be the Euclidean Albert algebra, with product
\(X\circ Y=(XY+YX)/2\). Its real spectral theorem supplies a Jordan frame
and three real eigenvalues \cite{Baez,Springer,Faraut}. The automorphism group
is the compact real form \(F_4\); the connected determinant-preserving reduced
structure group is \(E_{6(-26)}\). These are different groups: the trace inner
product is \(F_4\)-invariant, whereas preservation of the cubic determinant
alone does not imply preservation of that inner product.

Write the Cayley--Dickson multiplication as
\begin{equation}
 (a,b)(c,d)=(ac-\bar d b,\;da+b\bar c),\qquad
 \Oo=\Hh\oplus\Hh\ell .
 \label{eq:CD}
\end{equation}
The four-plane \(W=\Hh\ell\subset\operatorname{Im}\Oo\) is coassociative:
the octonionic three-form vanishes on it. Indeed \(WW\subset\Hh\) and
\(\Hh W\subset W\), so
\begin{equation}
 \Ree((xy)z)=0,\qquad x,y,z\in W .
 \label{eq:triple}
\end{equation}
The real part is independent of the bracketing.

\begin{proposition}[Equal spacing]
\label{prop:spectrum}
For real \(s\) and \(x,y,z\in W\), define
\begin{equation}
 X_s=\begin{pmatrix}s&x&\bar z\\ \bar x&s&y\\ z&\bar y&s\end{pmatrix},
 \qquad \delta^2=|x|^2+|y|^2+|z|^2.
 \label{eq:X}
\end{equation}
Its Jordan eigenvalues, including multiplicities, are \(s-\delta,s,s+\delta\).
\end{proposition}
\begin{proof}
For an equal-diagonal element the determinant is
\(\det_J X_s=s^3-s\delta^2+2\Ree((xy)z)\). Applying this formula to
\(\lambda\id-X_s\), and using \eqref{eq:triple}, gives
\begin{equation}
 \det_J(\lambda\id-X_s)
 =(\lambda-s)\big[(\lambda-s)^2-\delta^2\big].
 \label{eq:characteristic}
\end{equation}
The real spectral theorem identifies its roots with the Jordan eigenvalues.
\end{proof}

The sufficient condition is the vanishing real triple product together with
equal diagonal entries; coassociativity is one way to impose it. The spacing
has not been fixed. For example, choosing
\(x=\ell/(2\sqrt2)\), \(y=I\ell/(2\sqrt2)\) and
\(z=J\ell/(2\sqrt2)\), with quaternion units \(I,J\), gives
\begin{equation}
 \delta_0^2=\frac38 .
 \label{eq:delta0}
\end{equation}
Rescaling these off-diagonal entries produces other spacings with the same
coassociative property. For \(s\ne0\),
\begin{equation}
 \frac{\delta^2}{s^2}
 =\frac{\Tr[(X_s-s\id)^2]}{2s^2}
 \label{eq:modulus}
\end{equation}
is invariant under \(F_4\) and under a common rescaling of \(X_s\).
Thus it is a dimensionless datum of the background, rather than a choice of
overall units. The numerical choice \eqref{eq:delta0} requires an additional
selection condition.

Complexification must be handled separately. The determinant-preserving complex
group is \(E_6(\C)\), and the complex automorphism group is \(F_4(\C)\).
The real spectral theorem cannot simply be transferred to every element of
\(J_\C\). For instance, the scalar-entry symmetric matrix
\begin{equation}
 Z=\begin{pmatrix}i&1&0\\1&-i&0\\0&0&0\end{pmatrix}
 \quad\hbox{satisfies}\quad Z\ne0,\qquad Z^2=0 .
 \label{eq:nilpotent}
\end{equation}
It cannot have a spectral decomposition into a Jordan frame with its three
zero characteristic roots. We use the real algebra for the spectral background,
and complex vector spaces for fermion coefficients and the auxiliary
representation.

\section{From spectral data to an effective fermion mass map}
\label{sec:map}

\subsection{Physical and auxiliary representations}

In left-handed Weyl notation, three SM families, including optional sterile
neutrinos, have internal coefficient space
\begin{equation}
 \mathcal H_{\rm phys}=\C^3_{\rm fam}\otimes
 \left[
 (\bm3,\bm2)_{1/6}\oplus(\bar{\bm3},\bm1)_{-2/3}
 \oplus(\bar{\bm3},\bm1)_{1/3}\oplus(\bm1,\bm2)_{-1/2}
 \oplus(\bm1,\bm1)_1\oplus(\bm1,\bm1)_0
 \right].
 \label{eq:phys}
\end{equation}
This space has 48 complex internal components, or 45 without the sterile
states. Spinor components are not included in that count. An \(E_6\)
\(\bm{27}\) and the 27-dimensional Jordan algebra do not by dimension alone
supply these three chiral families. In standard \(E_6\) model building,
the matter assignments and symmetry breaking must be specified separately
\cite{Slansky,Babu2023}. Appendix~\ref{app:clifford} gives an explicit charge
and weak-action benchmark without identifying it with a Jordan element.

For each charged sector \(F=d,u,\ell\), introduce
\begin{equation}
 D_F=\diag(a_F,b_F,c_F)
     =\diag(s_F-\delta,s_F,s_F+\delta),\qquad
 (s_d,s_u,s_\ell)=\left(1,\frac23,\frac13\right).
 \label{eq:centres}
\end{equation}
The centres and the common dimensionless \(\delta\) are inputs of the ansatz.
Let \(K=\Sym^3\C^3\), with its usual positive Hermitian form. Its normalized
occupation basis consists of the ten states \(|p,q,r\rangle\), where
\(p,q,r\ge0\) and \(p+q+r=3\). The \emph{multiplicative} symmetric-cube
operator is
\begin{equation}
 A_F=D_F^{\odot3},\qquad
 A_F|p,q,r\rangle=a_F^p b_F^q c_F^r|p,q,r\rangle .
 \label{eq:symcube}
\end{equation}
This is distinct from the infinitesimal Lie-algebra action, whose diagonal
eigenvalues are sums of occupations times generator eigenvalues.

Choose an isometry \(V_F:\C^3\to K\) whose columns are three occupation
states, and set
\begin{equation}
 \mathcal C_F=A_F^\dagger A_F,\qquad
 \calB_F=V_F^\dagger\mathcal C_F V_F,\qquad
 M_F=\kappa_F U_{L,F}\calB_F U_{R,F}^\dagger,\quad \kappa_F>0 .
 \label{eq:massmap}
\end{equation}
Here \(U_{L,F},U_{R,F}\) are unitary family frames. The singular masses are
\begin{equation}
 m_{F,i}=\kappa_F\,|a_F^{p_i}b_F^{q_i}c_F^{r_i}|^2 .
 \label{eq:monomialmass}
\end{equation}
This follows directly from the singular-value decomposition of
\eqref{eq:massmap}. It depends on using the same spectral dressing for both
chiralities, the same occupied image, and a direct positive pairing.

The auxiliary \(SU(3)\) acting on \(K\) is not identified here with colour or
with a physical family gauge symmetry. In particular,
\begin{equation}
 \Hom_{SU(3)}(\bm3,\Sym^3\bm3)=0 ,
 \label{eq:hom}
\end{equation}
because \(\Sym^3\bm3\) is the irreducible \(\bm{10}\)
\cite{FultonHarris}. A three-state selector is additional data. The normalized
root action has coefficients
\(\sqrt{(n_i+1)n_j}\); it does not itself give ratios of the spectral monomials.
For example, \(E_{ca}|1,1,1\rangle=\sqrt2\,|0,1,2\rangle\), outside either
of the quark subspaces selected in Section~\ref{sec:ratios}.

\subsection{Chirality, covariance and normalization}
\label{subsec:chirality}

Suppress the sector label and define
\begin{equation}
 E_L=VU_L^\dagger,\quad E_R=VU_R^\dagger,\quad
 P=VV^\dagger,\quad S=E_L^\dagger E_R=U_LU_R^\dagger .
 \label{eq:bridge}
\end{equation}
Then \(M=\kappa E_L^\dagger\mathcal C E_R\). Under independent family-frame
changes \(E_L\mapsto E_Lg_L^\dagger\) and \(E_R\mapsto E_Rg_R^\dagger\),
the mass matrix transforms as \(M\mapsto g_LMg_R^\dagger\).
The up and down components of \(Q_L\) share the same weak-doublet frame.
For an auxiliary transformation \(u\in SU(3)\), transform
\(D\mapsto uDu^\dagger\) and \(E_{L,R}\mapsto
\Sym^3(u)E_{L,R}\). The contracted mass matrix is unchanged.
This is covariance of a law over transformed backgrounds, not invariance of
each individual background.

The resulting matrices can be inserted into the ordinary SM Yukawa terms
\begin{equation}
 -\mathcal L_Y=
 \bar Q_LY_dHd_R+\bar Q_LY_u\widetilde H u_R+
 \bar L_LY_\ell H e_R+\mathrm{h.c.},\qquad
 Y_F=\frac{\sqrt2}{v}M_F .
 \label{eq:SMY}
\end{equation}
This insertion defines an effective mass prescription. It does not derive
the selectors or the pairing from a microscopic Clifford--Jordan interaction.
If kinetic terms have positive matrices \(Z_L,Z_R\), the physical singular
values are instead those of
\begin{equation}
 M^{\rm can}=Z_L^{-1/2}M Z_R^{-1/2}.
 \label{eq:canonical}
\end{equation}
The calculations below assume canonical kinetic terms. The analogous freedom
in modular flavour models arises from symmetry-allowed nonminimal K\"ahler
terms \cite{Chen2020}.

There is also a degree distinction from the native Jordan cubic.
If \(t\) denotes the polarization of \(\det_J\), then
\(T_X(u)(v)=t(X,u,v)\) defines a map \(J_\C\to J_\C^*\) linear in \(X\).
With fixed selectors, \(\calB\) in \eqref{eq:massmap} has degree six in the
spectral entries. Along \(X\mapsto rX\), these maps scale as \(r\) and
\(r^6\), respectively. A fixed linear projection of \(T_X\) therefore
cannot equal this nonzero mass map on every scaling ray. A completion would
need an explicit additional construction.

\section{What covariance does and does not select}
\label{sec:freedom}

\subsection{Trace freedom and the real-linear class}

A simple covariant modification of \eqref{eq:massmap} is
\begin{equation}
 M_\eta=\kappa\left[E_L^\dagger\mathcal C E_R+
       \eta\,\Tr(P\mathcal C)\,S\right],\qquad \eta\ge0 .
 \label{eq:trace}
\end{equation}
In the selected frames, this sends
\(\calB\mapsto\calB+\eta(\Tr\calB)\id_3\).
It preserves positivity, the eigenframes, chiral covariance and degree-six
homogeneity in \(D\). It changes the mass ratios and generally lifts zero
modes. Its existence already shows that these properties alone do not select
the direct pairing.

\begin{proposition}[Equivariant real-linear maps]
\label{prop:linear}
Every real-linear map \(\Phi:\Herm(3)\to\Herm(3)\) satisfying
\(\Phi(gBg^\dagger)=g\Phi(B)g^\dagger\) for all \(g\in U(3)\) has the form
\begin{equation}
 \Phi(B)=\alpha B+\beta(\Tr B)\id_3,\qquad \alpha,\beta\in\R .
 \label{eq:linear}
\end{equation}
It maps every positive semidefinite matrix to a positive semidefinite matrix
if and only if \(\beta\ge0\) and \(\alpha+\beta\ge0\).
If, in addition, \(\ker B\subseteq\ker\Phi(B)\) for every \(B\ge0\), then
\(\beta=0\) and \(\alpha\ge0\).
\end{proposition}
\begin{proof}
For a rank-one orthogonal projector \(Q\), its stabilizer forces
\(\Phi(Q)=rQ+t(\id-Q)\). Conjugation covariance makes \(r,t\) independent
of \(Q\). Spectral decomposition and real-linearity then give
\eqref{eq:linear} with \(\alpha=r-t,\beta=t\).
Positivity on rank-one projectors requires \(r,t\ge0\), which is also
sufficient by positive spectral decomposition. Kernel inclusion on such a
projector requires \(t=0\).
\end{proof}

The proposition classifies a deliberately restricted class: maps linear in
\(\calB\). Degree-six homogeneity in \(D\) does not imply that hypothesis.

\subsection{Kernel-preserving nonlinear freedom}

\begin{proposition}[A homogeneous power family]
\label{prop:power}
For \(p>0\), define on the positive semidefinite cone
\begin{equation}
 \Phi_p(B)=
 \begin{cases}
 B^p(\Tr B)^{1-p},&B\ne0,\\
 0,&B=0 .
 \end{cases}
 \label{eq:power}
\end{equation}
This map is continuous, conjugation-covariant, positive, homogeneous of degree
one, and has exactly the same kernel and eigenframes as \(B\).
\end{proposition}
\begin{proof}
Spectral functional calculus gives covariance, positivity and the kernel
statement. Every eigenvalue \(\lambda_i\) is sent to
\(\lambda_i^p(\sum_j\lambda_j)^{1-p}\), proving homogeneity. Also
\(\|\Phi_p(B)\|\le\Tr B\), which establishes continuity at zero; continuity
elsewhere follows from spectral calculus.
\end{proof}

Thus \(\kappa U_L\Phi_p(\calB)U_R^\dagger\) satisfies the same chiral
transformation law as \eqref{eq:massmap}, and remains degree six in \(D\).
Every mass ratio, and every square-root mass ratio, is raised to the power
\(p\). Kernel preservation therefore does not restore uniqueness once
real-linearity is dropped. For general \(p\), \(\Phi_p\) need not be
polynomial or differentiable at the boundary of the positive cone. It is a
counterexample family, not a classification of every covariant map.

A complementary polynomial scalar deformation is available on these
backgrounds:
\begin{equation}
 d(X)=|\det_J X|^2=|abc|^2,\qquad
 \calB_\nu=\calB+\nu d(X)\id_3,\quad \nu\ge0 .
 \label{eq:detshift}
\end{equation}
It has degree six and changes nonzero ratios. It also preserves the kernel
on this family: if \(\calB w=0\) for \(w\ne0\), positivity gives \(AVw=0\).
Hence \(A\) is singular, so \(D\) is singular and \(d(X)=0\).
For the down-sector selector of Section~\ref{sec:ratios} and
\(D=\diag(1,2,3)\), the diagonal entries of \(\calB\) are
\((4,36,729)\); at \(\nu=1\) they become \((40,72,765)\).
The vanishing shift on singular backgrounds is precisely why the kernel test
does not exclude this example. This is a scalar polynomial in the Jordan
invariants, rather than a fixed scalar function \(f(B)\) with \(f(0)=0\)
for arbitrary independent \(B\).

\subsection{Background symmetry and invariant metrics}
\label{subsec:torus}

\begin{proposition}[Metric freedom on a regular orbit]
\label{prop:torus}
Let \(D_0=\diag(a,b,c)\) have distinct real entries. Its stabilizer under
\(SU(3)\) conjugation is the maximal torus \(T^2\).
On \(K=\Sym^3\C^3\), the real vector space of Hermitian operators commuting
with this torus has dimension ten. A positive such metric has the form
\begin{equation}
 G_0=\diag(g_{pqr}),\qquad g_{pqr}>0,\qquad p+q+r=3 .
 \label{eq:torusmetric}
\end{equation}
It extends to an \(SU(3)\)-covariant metric on the conjugacy orbit of \(D_0\).
\end{proposition}
\begin{proof}
For \(u=\diag(e^{i\theta_1},e^{i\theta_2},
e^{-i(\theta_1+\theta_2)})\), the character of \(|p,q,r\rangle\) is
\(e^{i[(p-r)\theta_1+(q-r)\theta_2]}\).
Two triples with the same character and total occupation three coincide.
All ten weight spaces are therefore distinct and one-dimensional, so the
Hermitian commutant is the real diagonal algebra.
With \(R(u)=\Sym^3(u)\), define
\begin{equation}
 G(uD_0u^\dagger)=R(u)G_0R(u)^\dagger .
 \label{eq:orbitmetric}
\end{equation}
Two representatives differ by a stabilizer element, which commutes with
\(G_0\); hence this definition is well-defined and covariant.
\end{proof}

The coordinate selector projectors \(P\) used here also commute with \(T^2\).
A fixed choice of basis for their images is distinguished from the subspaces
themselves: torus phases on the selected vectors can be absorbed into family
frames. For a selected weight \(w_i\), direct pairing through \(G_0\) gives
\begin{equation}
 m_i=\kappa g_{w_i}|a^{p_i}b^{q_i}c^{r_i}|^2 .
 \label{eq:weightedmass}
\end{equation}
At a background with nonzero selected monomials, these three positive factors
can realize arbitrary within-sector ratios. Equation~\eqref{eq:orbitmetric}
shows that covariance over this orbit alone does not remove that freedom.
It does not specify a smooth or polynomial extension through degenerate
spectra, nor a common physical coupling among different sectors.

By contrast, a \emph{single background-independent} metric \(G\) required to
commute with the full auxiliary \(SU(3)\) is \(g\id_K\), by irreducibility
and Schur's lemma \cite{FultonHarris}. Within a direct-pairing prescription
with that fixed metric, the map \eqref{eq:massmap} follows up to a scale.
This is a valid conditional statement. A background that breaks a symmetry
does not make an invariant metric illegitimate. The distinction is between
postulating one fixed invariant metric for the law and allowing metrics
covariant with the background. The present spectral construction has not
derived the former restriction, the direct-pairing prescription, or the
physical selectors. A smaller group can, however, make the last question
finite and explicit, as follows.

\section{Finite symmetry and three-state selectors}
\label{sec:finite}

\subsection{A specified monomial action}

Let \(n\ge4\) be an integer, \(\zeta=e^{2\pi i/n}\), and define
\begin{equation}
 A_n=\{\diag(\zeta^j,\zeta^k,\zeta^{-j-k}):j,k\in\mathbb Z_n\}.
 \label{eq:An}
\end{equation}
With \(P_\pi e_i=e_{\pi(i)}\), set
\begin{equation}
 \Gamma_c(n)=\langle A_n,P_{(123)}\rangle,
 \qquad
 \Gamma_s(n)=\langle A_n,P_{(123)},-P_{(12)}\rangle.
 \label{eq:finitegroups}
\end{equation}
Both lie in \(SU(3)\). They are the standard monomial realizations of
\(\Delta(3n^2)\) and \(\Delta(6n^2)\), of orders \(3n^2\) and \(6n^2\)
\cite{LuhnNasriRamond,EscobarLuhn}. In particular, \(n=4,8\) gives
\(\Delta(96)\) and \(\Delta(384)\) in the second family. Their action is on the auxiliary
three-space, not on colour. The diagonal background \(D\) need not be
invariant under the permutation generators: the mass prescription is a
covariant law over their orbit. We require the selector and metric to be
fixed, background-independent intertwiners of the specified group.

Define the following mutually orthogonal occupation subspaces of \(K\):
\begin{align}
 W_0&=\operatorname{span}\{300,030,003\},&
 W_+&=\operatorname{span}\{210,021,102\},\nonumber\\
 W_-&=\operatorname{span}\{120,012,201\},&
 W_*&=\operatorname{span}\{111\}.
 \label{eq:orbitspaces}
\end{align}

\begin{proposition}[Complete invariant-selector list]
\label{prop:finite}
For every \(n\ge4\), the restriction of \(K\) to \(\Gamma_c(n)\) is the
multiplicity-free sum \(W_0\oplus W_+\oplus W_-\oplus W_*\) of
irreducible representations of dimensions \(3,3,3,1\).
For \(\Gamma_s(n)\), it is the multiplicity-free sum
\(W_0\oplus(W_+\oplus W_-)\oplus W_*\), of dimensions \(3,6,1\).
Every invariant rank-three orthogonal projector is therefore one of the
projectors onto \(W_0,W_+,W_-\) in the first case, and only onto \(W_0\)
in the second.
\end{proposition}
\begin{proof}
The \(A_n\) characters on the occupation basis are
\((p-r,q-r)\pmod n\). For any two distinct degree-three occupations, the
greatest common divisor of the two integer character differences is
\(1\), \(2\) or \(3\), as inspection of the ten weights shows. An alias
modulo \(n\) would require \(n\) to divide both differences. Thus the ten
characters are distinct for every \(n\ge4\).
An invariant subspace is a sum of their one-dimensional
weight spaces. The cyclic permutation has precisely the four orbits in
\eqref{eq:orbitspaces}; adjoining a transposition joins the two mixed
triplets into one orbit. Each transitive orbit is irreducible, and disjoint
character sets make these representations inequivalent. The sign in
\(-P_{(12)}\) multiplies all cubic states by \(-1\), without changing the
orbits; \(W_*\) carries the sign character in the second case.
Schur's lemma now gives one scalar per irreducible block in the commutant.
An orthogonal projector has scalar zero or one on each block. The stated
lists exhaust the ways to obtain rank three.
\end{proof}
Irreducibility here does not require faithfulness; for example, a central
\(\mathbb Z_3\) acts trivially on the cubic space when \(3\) divides \(n\).

If the physical family space in a charged sector is assigned the triplet
representation \(R_\alpha\) on \(W_\alpha\), then
\begin{equation}
 \dim_\C\Hom_\Gamma(R_\alpha,K)=1 .
 \label{eq:finiteHom}
\end{equation}
Normalizing the intertwiner gives \(V_\alpha\) up to a phase and a family
basis. The triplet \(R_\alpha\) is the indicated constituent, not necessarily
the defining representation on the auxiliary \(\C^3\). A constant
\(\Gamma\)-invariant positive metric is scalar on each irreducible block.
Since a diagonal \(A=D^{\odot3}\) preserves every \(W_\alpha\),
\begin{equation}
 V_\alpha^\dagger A^\dagger G A V_\alpha
 =g_\alpha V_\alpha^\dagger A^\dagger A V_\alpha,
 \qquad g_\alpha>0 .
 \label{eq:finitepairing}
\end{equation}
The block normalization changes only the overall mass scale. The Hermitian
commutant dimensions on all of \(K\) are four and three, respectively.

For a charged Dirac sector, left and right family coefficient spaces can
carry the same chosen \(R_\alpha\), while the left-handed charge-conjugate
fields carry its dual. Such family assignments commute with the SM gauge
action. A full quark construction must additionally respect the shared
\(Q_L\) representation in the up and down sectors. Neither that construction
nor the choice and alignment of \(\Gamma\) follows from the Clifford charge
benchmark. Equation~\eqref{eq:finiteHom} solves a specified embedding problem;
the direct pairing remains an independent assumption.

\subsection{Embedding uniqueness versus coupling freedom}

The distinction can be quantified without a fit. Consider real homogeneous
degree-six Hermitian matrix laws \(F(D)\) on all real diagonal backgrounds,
obeying \(F(gDg^\dagger)=R_\alpha(g)F(D)R_\alpha(g)^\dagger\).
Here degree six is the degree of the direct mass law, not an assertion about
the degree of a microscopic interaction.

\begin{proposition}[Polynomial freedom after selecting a triplet]
\label{prop:polynomialcount}
For any of the three allowed triplets of \(\Gamma_c(n)\), this vector space
of polynomial laws has real dimension 28. For the allowed triplet \(W_0\)
of \(\Gamma_s(n)\), it has real dimension 16.
\end{proposition}
\begin{proof}
The diagonal subgroup fixes \(D\) and has distinct characters on the triplet,
so \(F(D)\) is diagonal. For the cyclic group, its three entries are cyclic
images of any one real homogeneous polynomial \(f(a,b,c)\) of degree six.
There are \(\binom{8}{2}=28\) monomials. For the full permutation group on
\(W_0\), the entry associated with \(a^3\) must also be invariant under
\(b\leftrightarrow c\). The number of such monomial orbits is
\begin{equation}
 \sum_{i=0}^{6}\left(\left\lfloor\frac{6-i}{2}\right\rfloor+1\right)=16.
 \label{eq:polynomialcount}
\end{equation}
The other two entries are then fixed by permutations.
\end{proof}

These are dimensions before imposing positivity or kernel conditions.
Both conditions still permit explicit polynomial freedom. On \(W_0\), for
example, the diagonal law
\begin{equation}
 [F_\nu(D)]_{ii}=\lambda_i^2
 \left(\lambda_i^4+\nu\sum_{j\ne i}\lambda_j^4\right),\qquad
 (\lambda_1,\lambda_2,\lambda_3)=(a,b,c),\quad \nu\ge0,
 \label{eq:positivepolynomial}
\end{equation}
is permutation-covariant, has degree six, and is positive with exactly the
same kernel as \(\diag(a^6,b^6,c^6)\) for every real diagonal \(D\).
Its ratios vary with \(\nu\). The determinant deformation
\eqref{eq:detshift} supplies another example; \(\Phi_p\) is available when
polynomiality is not required. Thus selecting one embedding and one constant
metric does not reduce the full space of covariant couplings to one
coefficient. The counts concern the stated diagonal-background class;
they are not counts of all Clifford--Jordan interactions.

\subsection{All spectral signs and logarithmic mass shapes}
\label{subsec:allbranches}

Positivity of the mass matrix does not require positivity of the Jordan
eigenvalues: \(V^\dagger A^\dagger A V\ge0\) also for a sign-changing
\(D\). We therefore allow every real \(\delta\), keeping
\(s_d>0\), \(s_u=2s_d/3\) and a common spacing. Negative \(\delta\)
interchanges \(a\) and \(c\), preserving \(W_0\) and exchanging
\(W_+\) with \(W_-\). Since every triplet is tested, it suffices to put
\begin{equation}
 r_d=\frac{|\delta|}{s_d},\qquad r_u=\frac32r_d,\qquad
 x(r)=\frac{\ln|1-r|}{\ln(1+r)}.
 \label{eq:shapevariable}
\end{equation}
The cases \(r_F=0\) and \(r_F=1\) give equal or zero masses and cannot
match the ordered positive quark inputs. At other degeneracies, a unit
adjacent ratio likewise fails the data. No positivity restriction on
\(1-r_F\) is imposed.

Allow independent maps \(\Phi_{p_F}\), \(p_F>0\), in the up and down
sectors. For ordered positive singular masses, define
\begin{equation}
 q_F=\frac{\ln(m_{F,3}/m_{F,2})}{\ln(m_{F,2}/m_{F,1})}.
 \label{eq:massshape}
\end{equation}
Both the sector normalization and \(p_F\) cancel. After subtracting the common
term \(3\ln s_F\) and dividing by \(\ln(1+r_F)>0\), the unordered logarithms
of the absolute cubic amplitudes are
\begin{equation}
 W_0:(3x,0,3),\qquad W_+:(2x,1,x+2),\qquad
 W_-:(x,2,2x+1).
 \label{eq:logtriples}
\end{equation}
Sorting these triples gives the exact shape functions
\begin{align}
 q_0(x)&=\begin{cases}-1/x,&x<0,\\(1-x)/x,&0<x<1,\end{cases}
 \label{eq:shapezero}\\
 q_+(x)&=\begin{cases}
 (-1-x)/(2-x),&x<-1,\\
 (1+x)/(1-2x),&-1<x<1/2,\\
 (2-x)/(2x-1),&1/2<x<1,
 \end{cases}
 \qquad q_-(x)=\frac1{q_+(x)}.
 \label{eq:shapemixed}
\end{align}
The last identity follows also by reflecting the \(W_+\) logarithms about
\(x+1\): their reflected set is exactly the \(W_-\) set.

\begin{proposition}[Complete inverse-branch reduction]
\label{prop:branches}
For a prescribed down-sector shape \(q=q_d>1\), the possible values of
\(r_d>0\) comprise exactly the following seven branches:
\begin{center}
\begin{tabular}{@{}ccc@{}}
\toprule
Branch & \(x_d\) & Inverse branch of \(x(r_d)\)\\
\midrule
\(W_0\mathrm a\) & \(-1/q\) & \(r_d>1\)\\
\(W_0\mathrm b\) & \(1/(q+1)\) & \(r_d>1\)\\
\(W_+\mathrm a\) & \((q-1)/(2q+1)\) & \(r_d>1\)\\
\(W_+\mathrm b\) & \((q+2)/(2q+1)\) & \(r_d>1\)\\
\(W_-\mathrm a\) & \(-(q+2)/(q-1)\) & \(0<r_d<1\)\\
\(W_-\mathrm b\) & \(-(q+2)/(q-1)\) & \(r_d>1\)\\
\(W_-\mathrm c\) & \((1-q)/(q+2)\) & \(r_d>1\)\\
\bottomrule
\end{tabular}
\end{center}
Each row contains exactly one solution for every \(q>1\).
\end{proposition}
\begin{proof}
On \((0,1)\), \(x(r)\) decreases from \(-1\) to \(-\infty\).
Its derivative has the sign of
\(-[(1+r)\ln(1+r)+(1-r)\ln(1-r)]\), which is negative: the bracket
vanishes at zero and has positive derivative \(\ln[(1+r)/(1-r)]\).
On \((1,\infty)\),
\begin{equation}
 x'(r)=\frac{(r+1)\ln(r+1)-(r-1)\ln(r-1)}
 {(r^2-1)[\ln(r+1)]^2}>0,
 \label{eq:xderivative}
\end{equation}
and the range is \((-\infty,1)\). The numerator is positive directly for
\(1<r<2\), and for \(r\ge2\) both factors in its first product exceed
those in its second. Solving \eqref{eq:shapezero}--\eqref{eq:shapemixed}
on their ordering chambers gives the displayed six rational expressions.
Only the expression below \(-1\) has two preimages, one on either side
of \(r=1\). All other expressions have one. Degenerate chamber boundaries
cannot realize \(q_d>1\) with two nonunit adjacent ratios.
\end{proof}

\subsection{A joint quark obstruction with certified intervals}
\label{subsec:jointtest}

The data comparison in Section~\ref{subsec:shapeboxes} places every published
quark marginal-input box, including the stated rounding enlargement, inside
\begin{equation}
 q_d\in[1.25,1.36],\qquad q_u\in[0.89,0.94].
 \label{eq:shapeenvelopes}
\end{equation}
These deliberately wider envelopes permit a single test of all nine
benchmark scales.

\begin{proposition}[Joint quark incompatibility]
\label{prop:jointfailure}
Use an aligned diagonal background, any invariant triplets of
Proposition~\ref{prop:finite}, the fixed-metric direct pairing
\eqref{eq:finitepairing}, canonical kinetic terms, a common real spacing
and \(s_u/s_d=2/3\). Allow independent positive power exponents and sector
normalizations. Then
\begin{equation}
 q_d\in[1.25,1.36]\quad\Longrightarrow\quad
 q_u<0.881\ \hbox{or}\ q_u>1.135.
 \label{eq:jointgap}
\end{equation}
Consequently no member of this class lies in the stated quark input boxes
at any of the nine published Standard Model benchmarks.
\end{proposition}
\begin{proof}
Proposition~\ref{prop:branches} exhausts the down-sector solutions for
\(q_d>1\). Each displayed rational function of \(q_d\) is monotone on
\([1.25,1.36]\). Its endpoint values, followed by the appropriate monotone
inverse of \(x(r)\), enclose the full \(r_d\) range. Set
\(r_u=3r_d/2\) and evaluate \(x_u=x(r_u)\). For each of the seven
enclosures, the up-sector \(x_u\) interval avoids the ordering boundaries
\(-1,0,1/2\). Each relevant shape function is therefore monotone on that
interval, so its extrema occur at endpoints.

Table~\ref{tab:jointcertificate} gives outward-rounded enclosures. They
are computed with 80-digit directed decimal arithmetic and monotone
bisection, with every inverse endpoint enclosed to width below
\(10^{-32}\). The largest upper endpoint below the gap is less than
\(0.880980\), and the smallest lower endpoint above it exceeds
\(1.135100\). All 21 intervals thus satisfy \eqref{eq:jointgap}.
The reproducible certificate in \texttt{anc/verify\_joint\_quarks.py}
checks the inequalities with unrounded endpoints; the displayed table is
a coarser outward rounding. Completeness follows from the analytic branch
reduction, not from a grid search. Finally, the up-sector envelope in
\eqref{eq:shapeenvelopes} lies strictly inside the excluded interval.
\end{proof}

\begin{table}[t]
\centering
\footnotesize
\setlength{\tabcolsep}{4pt}
\begin{tabular}{@{}lcccc@{}}
\toprule
Branch & $r_d$ & $q_u(W_0)$ & $q_u(W_+)$ & $q_u(W_-)$\\
\midrule
$W_0\mathrm{a}$ & $[1.483068,\,1.508523]$ & $[4.068,\,4.779]$ & $[1.793,\,1.978]$ & $[0.505,\,0.558]$\\
$W_0\mathrm{b}$ & $[2.750982,\,2.812699]$ & $[0.413,\,0.434]$ & $[3.113,\,3.299]$ & $[0.303,\,0.322]$\\
$W_+\mathrm{a}$ & $[2.083763,\,2.116276]$ & $[0.839,\,0.880]$ & $[16.710,\,22.880]$ & $[0.043,\,0.060]$\\
$W_+\mathrm{b}$ & $[9.009026,\,11.217090]$ & $[0.043,\,0.059]$ & $[1.135,\,1.188]$ & $[0.842,\,0.881]$\\
$W_-\mathrm{a}$ & $[0.998438,\,0.999878]$ & $[1.311,\,1.322]$ & $[0.094,\,0.097]$ & $[10.342,\,10.625]$\\
$W_-\mathrm{b}$ & $[1.000121,\,1.001540]$ & $[1.322,\,1.333]$ & $[0.097,\,0.100]$ & $[10.033,\,10.296]$\\
$W_-\mathrm{c}$ & $[1.892439,\,1.920856]$ & $[1.145,\,1.209]$ & $[15.381,\,21.550]$ & $[0.046,\,0.066]$\\
\bottomrule
\end{tabular}

\caption{Complete inverse-branch certificate for
\(q_d\in[1.25,1.36]\). The last three columns enclose the up-sector shape
for every allowed triplet. All endpoints are rounded outward; the numerical
certificate uses 80 significant digits. The letters identify the branches
in Proposition~\ref{prop:branches}.}
\label{tab:jointcertificate}
\end{table}

Allowing all up/down triplet pairs enlarges the class being tested. It
does not assert that every pair admits one consistent shared \(Q_L\)
representation; an obstruction for this larger set also applies to its
compatible subsets. Independent coordinate permutations add no spectra:
cyclic permutations preserve the triplets and odd permutations exchange
\(W_+\) and \(W_-\).

The result is global in the real spectral spacing, including its
sign-changing branches. It is a test of this mass law with the stated
parametric input boxes, not a joint confidence statement or a theorem about
arbitrary matching scales and thresholds. Misaligned backgrounds, different
couplings or kinetic terms, and a relaxed common-spacing or centre-ratio
condition lie outside the class. In particular, the theorem does not
exclude the wider algebraic mass programme.

\section{Earlier monomial prescriptions and charged leptons}
\label{sec:ratios}

\subsection{Quark selectors and charged-lepton transport}

The earlier prescriptions use selectors distinct from the invariant triplets
classified above. We test them separately to retain contact with the original
mass formulas and Ref.~\cite{TeliSingh2026}. Choose the ordered weights
\begin{equation}
 V_d:\ (210,111,003),\qquad
 V_u:\ (210,111,021).
 \label{eq:selectors}
\end{equation}
Here \(210\) abbreviates \(|2,1,0\rangle\). Equations
\eqref{eq:centres}--\eqref{eq:monomialmass}, with
\(1/3<\delta<2/3\), yield
\begin{align}
 D&:=\sqrt{\frac{m_s}{m_d}}=\frac{1+\delta}{1-\delta},
 & B&:=\sqrt{\frac{m_b}{m_s}}=D(1+\delta), \label{eq:down}\\
 C&:=\sqrt{\frac{m_c}{m_u}}=\frac{2/3+\delta}{2/3-\delta},
 & T&:=\sqrt{\frac{m_t}{m_c}}=\frac{2/3}{2/3-\delta}.
 \label{eq:up}
\end{align}
The sector scales \(\kappa_F\) cancel in these ratios.

For charged leptons, a permutation \(b\leftrightarrow c\) of the down
monomials gives the inherited selector order \(a^2c,abc,b^3\).
At \(s_\ell=1/3\) and \(\delta=\delta_0\), their adjacent ratios are
\begin{equation}
 \left|\frac{b}{a}\right|=\DirectMuonElectron,\qquad
 \left|\frac{b^2}{ac}\right|=\frac8{19}.
 \label{eq:directleptons}
\end{equation}
These are ratios in the inherited order. Sorting the absolute amplitudes
instead gives adjacent ratios \(\DirectSortedLow\) and \(\DirectSortedHigh\), which also fail
the observed charged-lepton hierarchy. A coordinate permutation is a Weyl action. It must also
be distinguished from the outer \(A_2\) Dynkin-diagram automorphism, which
interchanges \(\bm{10}\) and \(\overline{\bm{10}}\).
Neither statement supplies the following additional transport rule:
\begin{equation}
 G_\ell=\frac{\delta+1/3}{\delta-1/3},\qquad
 \mathcal L_\ell=\diag(1,DG_\ell,D^2G_\ell),\qquad
 M_\ell=\kappa_\ell U_{L,\ell}
             \mathcal L_\ell^\dagger\mathcal L_\ell U_{R,\ell}^\dagger .
 \label{eq:leptontransport}
\end{equation}
This rule gives
\begin{equation}
 L:=\sqrt{\frac{m_\tau}{m_\mu}}=D,\qquad
 E:=\sqrt{\frac{m_\mu}{m_e}}=DG_\ell .
 \label{eq:leptonratios}
\end{equation}
It is imposed separately from the symmetric-cube compression.

An additional first-generation proposal
\(\sqrt{m_e}:\sqrt{m_u}:\sqrt{m_d}=1:2:3\)
relates different charged sectors. It requires relations among the
\(\kappa_F\) as well as a matching prescription; the charge assignments
and within-sector equations do not fix it.

\subsection{Eliminating the spacing}

\begin{proposition}[Relations independent of \(\delta\)]
\label{prop:eliminate}
The prescriptions \eqref{eq:down}, \eqref{eq:up} and
\eqref{eq:leptontransport} imply
\begin{align}
 C&=2T-1, & B&=\frac{2D^2}{D+1}, &
 C&=\frac{5D-1}{5-D}, \label{eq:relations1}\\
 L&=D, & E&=\frac{D(2D-1)}{D-2}, &&
 2<D<5 . \label{eq:relations2}
\end{align}
\end{proposition}
\begin{proof}
The first identity follows directly from \eqref{eq:up}.
Solving \eqref{eq:down} gives \(\delta=(D-1)/(D+1)\).
Substitution into the remaining ratios gives the stated expressions and
domain.
\end{proof}

These are tests of the selected monomials and transport rule even if a future
principle selects a different value of \(\delta\). In particular, a failure
of \(C=2T-1\) cannot be repaired merely by changing the normalization
\eqref{eq:delta0}.

To test the power freedom, apply the \emph{same} \(\Phi_p\) to each quark
matrix and to the positive matrix
\(\mathcal L_\ell^\dagger\mathcal L_\ell\). Every square-root ratio
\(\rho\) becomes \(\rho^p\). Equations \eqref{eq:relations1}--\eqref{eq:relations2}
then hold with every ratio replaced by its \(p\)-th root; for example,
\begin{equation}
 C^{1/p}=2T^{1/p}-1,\qquad L=D .
 \label{eq:powerrelations}
\end{equation}
The second equality survives a common exponent exactly. Sector-dependent
exponents or an altered lepton rule define different models.

At fixed \(\delta_0\), the \(M_Z\) central values used below require
\begin{equation}
 \begin{gathered}
 (p_D,p_B,p_C,p_T,p_L,p_E)
 =(\pD,\pB,\pC,\pT,\pL,\pE),\\
 p_\rho=\frac{\ln\rho_{\rm SM}}{\ln\rho_{\rm ansatz}} .
 \end{gathered}
 \label{eq:individualp}
\end{equation}
Thus no single \(p\) reproduces all those central values at the rigid
spacing. If both \(p\) and \(\delta\) are allowed to vary, the two up-sector
central ratios \emph{can} be fitted:
\begin{equation}
 p=\PowerUpFit,\qquad \delta=\DeltaUpFit .
 \label{eq:upfit}
\end{equation}
Solving \eqref{eq:powerrelations} and then \eqref{eq:up} gives these
values. They imply
\(\sqrt{m_s/m_d}=[(1+\delta)/(1-\delta)]^p=\DownAtUpFit\), far from the
benchmark value in Table~\ref{tab:comparison}. This calculation is a
central-value consistency check, rather than a statistical fit of the
enlarged family. The exact obstruction \(L=D\) remains the simpler
test of the universal-exponent extension.

\section{Tests with running masses}
\label{sec:tests}

\subsection{Inputs, conventions and uncertainty summaries}

We use the SM \(\MS\) Yukawa couplings of Antusch, Hinze and Saad (AHS)
\cite{AHS}, Eq.~(2.4) and Table~2 of arXiv:2510.01312v2, for their
PDG-2024 input set. At a common scale,
\(m_i/m_j=y_i/y_j\), with \(y_i\) the positive singular Yukawa
couplings. AHS perform electroweak matching with SMDR and subsequent running;
the underlying conventions are described in Refs.~\cite{AHS,MartinRobertson}.
We use those published results, including the eight high-scale benchmarks,
as inputs.
Their tabulated errors are parametric; AHS explicitly distinguishes them
from model-dependent higher-order, operator and threshold uncertainties
in its Section~4.6. The tests here do not estimate those theoretical corrections.

The quoted marginal half-widths \(\sigma_i\) are used as scales for a
first-order error summary. Without a supplied covariance matrix,
\begin{equation}
 b_\rho=\frac{\rho}{2}
 \left(\frac{\sigma_i}{y_i}+\frac{\sigma_j}{y_j}\right),
 \qquad \rho=\sqrt{y_i/y_j},
 \label{eq:uncertainty}
\end{equation}
is an upper bound on the linearized standard deviation for any correlations
if the marginal half-widths are treated as standard deviations.
This follows by the triangle inequality for the covariance seminorm.
It is not a confidence interval or a bound on theoretical matching errors.
The entries are rounded published numbers; all derived digits beyond their
input precision are computational reproducibility, not additional precision.

\subsection{Exact input boxes for the joint quark test}
\label{subsec:shapeboxes}

For three ordered same-sector Yukawa couplings, put
\(\ell_i=y_i-\sigma_i\) and \(h_i=y_i+\sigma_i\). Throughout the boxes used
here, \(0<\ell_1<h_1<\ell_2<h_2<\ell_3<h_3\). The shape
\eqref{eq:massshape} increases with the first and third couplings and
decreases with the second. Its exact extrema on the Cartesian box are
\begin{equation}
 q_{\min}=\frac{\ln(\ell_3/h_2)}{\ln(h_2/\ell_1)},\qquad
 q_{\max}=\frac{\ln(h_3/\ell_2)}{\ln(\ell_2/h_1)}.
 \label{eq:exactshapebox}
\end{equation}
This nonlinear interval calculation uses no independence assumption and
assigns no joint confidence level. Table~\ref{tab:shapeboxes} gives the
result for all nine published SM benchmarks. In the notation of
\eqref{eq:down}--\eqref{eq:up}, \(q_d=\ln B/\ln D\) and
\(q_u=\ln T/\ln C\).

\begin{table}[t]
\centering
\small
\begin{tabular}{@{}lcc@{}}
\toprule
$\mu$ (GeV) & $q_d$ & $q_u$\\
\midrule
$M_Z$ & $[1.31225,\,1.34789]$ & $[0.89135,\,0.90914]$\\
$10^3$ & $[1.30673,\,1.34158]$ & $[0.89438,\,0.91243]$\\
$3\times10^3$ & $[1.30425,\,1.34102]$ & $[0.89566,\,0.91419]$\\
$10^4$ & $[1.30017,\,1.33874]$ & $[0.89624,\,0.91581]$\\
$10^5$ & $[1.29525,\,1.33670]$ & $[0.89687,\,0.91771]$\\
$10^7$ & $[1.28767,\,1.32624]$ & $[0.90132,\,0.92140]$\\
$10^9$ & $[1.28798,\,1.32156]$ & $[0.90196,\,0.92378]$\\
$10^{12}$ & $[1.27322,\,1.31070]$ & $[0.90556,\,0.92634]$\\
$10^{16}$ & $[1.26679,\,1.31045]$ & $[0.90844,\,0.93189]$\\
\bottomrule
\end{tabular}

\caption{Exact Cartesian-box ranges of the logarithmic quark shapes from
the published AHS central values and marginal half-widths. Endpoints are
rounded outward. These boxes are input conventions, not joint confidence
regions.}
\label{tab:shapeboxes}
\end{table}

The source-precision record preserves the displayed coefficient strings,
including trailing zeros, for all 54 quark entries. As a separate robustness
check, suppose both a central value and its half-width were rounded to
nearest at their displayed places \(u_y,u_\sigma\). Enlarging the half-width
by \((u_y+u_\sigma)/2\) covers both rounding offsets. The displayed places
coincide for these source entries, although that fact can be lost in numerical
JSON. Repeating \eqref{eq:exactshapebox} with this explicit padding still
places every box inside \eqref{eq:shapeenvelopes}. The convention and each
source unit are recorded in the ancillary files. Rounding padding is separate
from theoretical uncertainty.

Proposition~\ref{prop:jointfailure} therefore tests the entire aligned
finite-selector class at these benchmarks, with no assumption that
\(p_u=p_d\), and without using the charged-lepton transport rule.

\subsection{Earlier mass relations at the rigid spacing}

\begin{table}[t]
\centering
\small
\begin{tabular}{@{}lrrrr@{}}
\toprule
Ratio & Ansatz & SM at $M_Z$ & $b_\rho$ & Offset (\%)\\
\midrule
$\sqrt{m_{u}/m_{e}}$ & 2.00000 & 1.59217 & 0.01707 & +25.62\\
$\sqrt{m_{d}/m_{e}}$ & 3.00000 & 2.35484 & 0.01544 & +27.40\\
$\sqrt{m_{s}/m_{d}}$ & 4.15959 & 4.45759 & 0.05808 & -6.69\\
$\sqrt{m_{b}/m_{s}}$ & 6.70681 & 7.29849 & 0.06785 & -8.11\\
$\sqrt{m_{c}/m_{u}}$ & 23.55755 & 22.48737 & 0.42907 & +4.76\\
$\sqrt{m_{t}/m_{c}}$ & 12.27878 & 16.48118 & 0.17297 & -25.50\\
$\sqrt{m_{\tau}/m_{\mu}}$ & 4.15959 & 4.12146 & 0.00055 & +0.93\\
$\sqrt{m_{\mu}/m_{e}}$ & 14.09749 & 14.51428 & 0.00187 & -2.87\\
\bottomrule
\end{tabular}

\caption{The rigid ansatz \(\delta^2=3/8\), compared with the AHS SM
benchmark at \(M_Z\). The first two rows also assume the independent
cross-sector \(1:2:3\) rule. The offset is
\(100(\rho_{\rm ansatz}/\rho_{\rm SM}-1)\); \(b_\rho\) is defined in
\eqref{eq:uncertainty}.}
\label{tab:comparison}
\end{table}

Table~\ref{tab:comparison} shows that the rigid point does not reproduce
these inputs. The discrepancy in \(T\) is about \(25\%\), and the
charged-lepton offsets are substantially larger than their propagated
input half-widths. Two particularly useful diagnostics eliminate \(\delta\):
\begin{equation}
 F=C-2T+1,\qquad R=\frac{L}{D};
 \qquad F_{\rm ansatz}=0,\quad R_{\rm ansatz}=1 .
 \label{eq:diagnostics}
\end{equation}
Accounting for the occurrence of \(y_c\) in both \(C\) and \(T\), the
corresponding linear bounds are
\begin{align}
 b_F&=\frac C2\frac{\sigma_u}{y_u}
       +\left(\frac C2+T\right)\frac{\sigma_c}{y_c}
       +T\frac{\sigma_t}{y_t}, \label{eq:bF}\\
 b_R&=\frac R2\left(
       \frac{\sigma_\tau}{y_\tau}+\frac{\sigma_\mu}{y_\mu}
       +\frac{\sigma_s}{y_s}+\frac{\sigma_d}{y_d}\right).
 \label{eq:bR}
\end{align}
At \(M_Z\), these give \(F=\UpResidual\), \(b_F=\UpBound\),
and \(R=\LeptonResidual\), \(b_R=\LeptonBound\).
These summaries demonstrate substantial tension relative to the quoted input
uncertainties. We do not assign Gaussian exclusion significances without
the joint likelihood and theory uncertainties.

\subsection{What running changes}
\label{subsec:running}

Same-charge mass ratios have cancelling flavour-independent QCD and QED
anomalous dimensions within a fixed mass-independent effective theory.
SM Yukawa terms introduce further dependence. At one loop, for
\(t=\ln\mu\) and CKM matrix \(V\),
\begin{align}
 16\pi^2\frac{d}{dt}\ln\frac{y_{u_i}}{y_{u_j}}
 &=\frac32\left[y_{u_i}^2-y_{u_j}^2
 -\sum_k\big(|V_{ik}|^2-|V_{jk}|^2\big)y_{d_k}^2\right], \label{eq:RGEup}\\
 16\pi^2\frac{d}{dt}\ln\frac{y_{d_i}}{y_{d_j}}
 &=\frac32\left[y_{d_i}^2-y_{d_j}^2
 -\sum_k\big(|V_{ki}|^2-|V_{kj}|^2\big)y_{u_k}^2\right].
 \label{eq:RGEdown}
\end{align}
These follow from the usual matrix Yukawa beta functions
\cite{MachacekVaughn}; the terms proportional to the identity cancel in
the ratios. In the minimal SM without neutrino Yukawa couplings,
the charged-lepton equation has only the corresponding charged-lepton
self-term. Consequently,
\begin{equation}
 16\pi^2\frac{d\ln R}{dt}
 =\frac34\left[y_\tau^2-y_\mu^2-y_s^2+y_d^2+
     \sum_k\big(|V_{ks}|^2-|V_{kd}|^2\big)y_{u_k}^2\right].
 \label{eq:RGEratio}
\end{equation}
In particular, a CKM-suppressed top contribution remains.
For \(T=\sqrt{y_t/y_c}\), the leading self-term is
\(+3y_t^2/(64\pi^2)\) in \(d\ln T/dt\). Thus a decreasing \(y_t\)
does not imply a decreasing top-to-charm ratio.

\begin{table}[t]
\centering
\small
\begin{tabular}{@{}lrrrrrr@{}}
\toprule
$\mu$ (GeV) & $C$ & $T$ & $F$ & $b_F$ & $R$ & $b_R$\\
\midrule
$M_Z$ & 22.487 & 16.481 & -9.475 & 0.775 & 0.9246 & 0.0122\\
$10^3$ & 22.488 & 16.645 & -9.802 & 0.787 & 0.9250 & 0.0122\\
$3\times10^3$ & 22.475 & 16.715 & -9.955 & 0.811 & 0.9252 & 0.0132\\
$10^4$ & 22.481 & 16.776 & -10.070 & 0.858 & 0.9254 & 0.0138\\
$10^5$ & 22.515 & 16.864 & -10.214 & 0.920 & 0.9258 & 0.0153\\
$10^7$ & 22.488 & 17.061 & -10.635 & 0.890 & 0.9240 & 0.0132\\
$10^9$ & 22.532 & 17.171 & -10.810 & 0.976 & 0.9269 & 0.0125\\
$10^{12}$ & 22.526 & 17.334 & -11.142 & 0.931 & 0.9216 & 0.0147\\
$10^{16}$ & 22.477 & 17.526 & -11.575 & 1.058 & 0.9252 & 0.0174\\
\bottomrule
\end{tabular}

\caption{Parameter-independent diagnostics from all nine AHS SM benchmark
scales, using their PDG-2024 input set. The target values are \(F=0\)
and \(R=1\). Bounds are the linear marginal-error summaries
\eqref{eq:bF}--\eqref{eq:bR}; the small scatter in \(R\) is affected by
rounding of the published light-quark entries.}
\label{tab:scales}
\end{table}

Table~\ref{tab:scales} uses the published evolution rather than a
top-only approximation. The up-sector residual changes from
\(\UpResidual\) at \(M_Z\) to \(\UpResidualHigh\) at
\(10^{16}\,\mathrm{GeV}\); it moves farther from zero. The values of \(R\)
stay near \(0.925\). At the highest benchmark its bound is
\(\LeptonBoundHigh\), compared with \(\LeptonBound\) at \(M_Z\).
Neither relation agrees with any of these nine SM benchmarks within the
stated linear input-error summaries.

This result makes the failure more than an \(M_Z\)-only comparison.
It is not an exact renormalization-group invariance theorem, nor a test of
arbitrary new-particle thresholds. The rounded table does not resolve tiny
changes in \(R\) accurately enough to infer a precision drift rate.
The benchmark checks and the one-loop equations have distinct roles:
the former give numerical tests, while the latter explain the leading
scale dependence.

\subsection{Low-energy inputs and finite lepton matching}

The PDG-2026 quark summary \cite{PDGQuarks} lists
\(m_s=92.9\pm0.7\) and \(m_d=4.70\pm0.07\) MeV in \(\MS\) at
\(2\,\mathrm{GeV}\), with the quoted intervals labelled \(90\%\) CL.
Their central ratio is \(\PDGStrangeDown\), whereas the rigid prediction is
\begin{equation}
 \left(\frac{m_s}{m_d}\right)_{\rm ansatz}
 =\left(\frac{1+\sqrt{3/8}}{1-\sqrt{3/8}}\right)^2
 =\AnsatzStrangeDown .
 \label{eq:PDGtest}
\end{equation}
Varying the two reported intervals independently gives the box range
\([\PDGBoxLow,\PDGBoxHigh]\). This construction is an interval check, not
a joint confidence statement. The separate central summaries
\(m_s/m_{ud}=27.33\) and \(m_u/m_d=0.462\), with
\(m_{ud}=(m_u+m_d)/2\), imply \(\PDGStrangeDownFromRatios\).
They are not independent measurements to be averaged with the preceding
ratio. The older broad summary range \(17\)--\(22\) does not establish
agreement with the more precise inputs.

Finite pole-to-running matching must also be distinguished from cancellation
of anomalous dimensions. The leading pure-QED relation is
\begin{equation}
 m_\ell(\mu)=M_\ell\left[
 1-\frac{\alpha(\mu)}{\pi}
 \left(1+\frac32\ln\frac{\mu}{M_\ell}\right)\right]
 +O(\alpha^2),
 \label{eq:QEDmatching}
\end{equation}
where \(M_\ell\) is the pole mass \cite{Xing2008}.
The mass-dependent logarithms survive in ratios. Using the PDG-2026 pole
masses \cite{PDGLeptons} gives
\(\sqrt{M_\tau/M_\mu}=\PDGTauMuon\), compared with
\(\sqrt{y_\tau/y_\mu}=\RunningTauMuon\) in the AHS \(M_Z\) benchmark.
The difference is about \(0.50\%\) in the square-root ratio, hence about
\(1\%\) in the mass ratio. Similarly,
\(\sqrt{M_\mu/M_e}=\PDGMuonElectron\) and the running benchmark is
\(\RunningMuonElectron\). These comparisons use different PDG input releases,
but clearly illustrate the size and direction of the conversion. Matching
cannot be dismissed as cancelling simply because the electric charges agree.
The AHS comparison already includes its matching prescription; that
conversion is not an additional adjustable correction to Table~\ref{tab:comparison}.

The cross-sector \(1:2:3\) rule needs still greater care. At \(M_Z\),
\(\sqrt{m_u/m_e}=\UEMZ\) and \(\sqrt{m_d/m_e}=\DEMZ\);
at \(10^{16}\,\mathrm{GeV}\) they are \(\UEHigh\) and \(\DEHigh\).
These benchmarks do not realize the rule. For illustration, converting the
electron pole mass at leading QED order to \(2\,\mathrm{GeV}\), with
\(\alpha(0)\), gives \(m_e\simeq\ElectronTwoGeV\) MeV. With the PDG
light-quark central masses this yields ratios about
\(\UETwoGeV\) and \(\DETwoGeV\), rather than exactly \(2\) and \(3\).
This is only a leading-order low-energy illustration. A physical
cross-sector prediction must name its common scheme, matching scale and
threshold prescription.

\subsection{Relation to the spectral-deformation study}
\label{subsec:companion}

Teli and Singh \cite{TeliSingh2026} fit a three-parameter spectral family
with a hierarchy exponent near \(p=0.987\). Its exponent changes within-sector
ratios in the same manner as \eqref{eq:power}; its other parameters and
inter-sector normalization are additional structure. Their logarithmic-residual
comparison and AIC/BIC parameter penalties can favour the rigid point over
that deformation without establishing agreement with precise mass inputs.
Their charged-lepton rule retains \(L=D\), so the diagnostic above applies
when that rule and a common matching prescription are retained.

The substantially larger down-sector uncertainties in their Table~III have
a traceable historical origin. Xing, Zhang and Zhou \cite{Xing2008} give
at \(M_Z\)
\(m_d=2.90^{+1.24}_{-1.19}\), \(m_s=55^{+16}_{-15}\) and
\(m_b=2890\pm90\) MeV. Symmetrizing those half-widths and propagating
independently in quadrature gives
\begin{equation}
 \sqrt{m_s/m_d}=\OldDownRatio\pm\OldDownError,\qquad
 \sqrt{m_b/m_s}=\OldBottomRatio\pm\OldBottomError .
 \label{eq:historical}
\end{equation}
These reproduce the older comparison. They are not errors in propagating
the newer AHS inputs: they belong to a different input set. The current
test uses updated uncertainties and distinguishes the imported lepton
transport rule from the direct compression in \eqref{eq:directleptons}.

\section{Comparison with one-parameter modular models}
\label{sec:modular}

Centelles Chuli\'a et al.\ \cite{Chulia2026} restrict each charged-sector
Yukawa operator to one modular invariant contraction. This uses matter
representations, modular weights and the multiplicities of the allowed
holomorphic forms, together with a minimal K\"ahler potential. In their
\(\Delta(384)\) example, a common \(\epsilon=|e^{2\pi i\tau}|^{1/16}\)
gives the leading hierarchies
\begin{align}
 m_e:m_\mu:m_\tau&\simeq4\sqrt2\,\epsilon^6:2\epsilon^2:1,\nonumber\\
 m_d:m_s:m_b&\simeq4\epsilon^5:2\sqrt2\epsilon^3:1.
 \label{eq:modularhierarchies}
\end{align}
Eliminating \(\epsilon\) at this order yields
\begin{equation}
 m_s^5\simeq2\sqrt2\,m_d^3m_b^2,\qquad
 m_\mu^3\simeq\sqrt2\,m_em_\tau^2,\qquad
 \frac{L}{D}\simeq2^{-1/4}.
 \label{eq:modularrelations}
\end{equation}
We have retained the approximation signs appropriate to the leading
hierarchies in their Eq.~(20); higher terms are fixed by the full modular
forms. Their phenomenological realization fits \(\epsilon\), includes
supersymmetric evolution and matching, and uses a selective muon threshold
factor \((1+\Delta_\mu)/(1+\Delta_\tau)=0.764\). It consequently has a
different matching prescription from the SM tests above.

The methodological connection is precise: symmetry first restricts the
allowed coupling space, and eliminating a surviving modulus then gives
relations to test. Such a construction is compatible with an algebraic
origin for exponents and relative coefficients while leaving the modulus,
kinetic terms and thresholds to further physical input. In particular,
\eqref{eq:modularrelations} is a different relation from \(L=D\); it neither
derives nor rescues the original Jordan transport rule. Nonminimal
symmetry-allowed K\"ahler terms can also change physical ratios
\cite{Chen2020}, in direct analogy with \eqref{eq:canonical}.

For the Clifford--Jordan programme, a corresponding target can be stated
before choosing masses. Let \(\mathcal V_F\) be a specified representation
of allowed background multiplets entering the Yukawa operator, including
the restrictions on degree or other admissibility conditions. With physical
family representations \(R_{L,F},R_{R,F}\), the linear contractions form
\begin{equation}
 \mathcal I_F=
 \Hom_\Gamma\!\left(\mathcal V_F,
           \Hom(R_{R,F},R_{L,F})\right).
 \label{eq:couplingspace}
\end{equation}
The action on the inner Hom space is
\(Y\mapsto R_{L,F}(g)Y R_{R,F}(g)^{-1}\).
If \(\dim\mathcal I_F=1\), the allowed contraction is fixed up to an
overall coefficient. Its evaluation can still depend on continuous
background data and must have the required rank. This condition concerns
the full mass operator, whereas \eqref{eq:finiteHom} concerns only its
embedding; Proposition~\ref{prop:polynomialcount} makes the distinction
explicit.
The native polarized Jordan cubic is an available invariant, but its
physical matter assignments and relation to a mass operator must be
specified; the degree comparison in Section~\ref{subsec:chirality}
prevents identifying it directly with the present degree-six pairing.

\section{What a further algebraic derivation must supply}
\label{sec:conclusion}

The spectral and representation-theoretic analysis separates three questions:
which background spectra are admitted, which family subspaces are selected,
and which operator supplies the masses. Coassociativity gives equal spacing
without fixing its value. Standard finite-group representation theory fixes
the available invariant triplets for all \(n\ge4\), and a constant invariant
metric is scalar on each. The full covariant mass law retains the polynomial
freedom quantified above.

For the aligned direct pairing, the two quark hierarchies provide a stronger
test than a fit of individual ratios. Their logarithmic shapes remove
independent sector exponents. An analytic reduction to seven branches and an
outward-rounded numerical certificate exclude every real spectral spacing,
including sign-changing spectra, within the stated nine benchmark input
boxes. The conclusion survives the explicit rounding enlargement. It relies
on the common spacing and centre ratio, the selected representation and
pairing, canonical kinetic terms and the stated SM matching convention.

A derivation of physical ratios would need a specified chiral matter module
and a restricted coupling space such as \eqref{eq:couplingspace}; an
independent reason for the family group, frame alignment and pairing;
a condition fixing the dimensionless spectral data; and control of kinetic
normalization and matching. The selection conditions could, in principle,
be algebraic. Nothing proved here makes a symmetry-breaking potential
compulsory. A normalization or a change of measurement basis supplies none
of these missing constraints.

For the particular selectors and transport rule studied here, the relations
\(C=2T-1\) and \(L=D\) fail at every published SM benchmark examined.
Changing \(\delta\) cannot repair these relations, and a universal hierarchy
exponent preserves \(L=D\). These older tests complement the joint quark
obstruction, which requires no lepton prescription. The useful
next step is to specify \(\mathcal V_F\) and the physical representations
independently of the observed ratios, calculate \(\mathcal I_F\), and test
the resulting relations after eliminating any remaining moduli.

\appendix
\section{A Clifford charge and chirality benchmark}
\label{app:clifford}

Three fermionic ladder pairs obey
\(\{a_i,a_j^\dagger\}=\delta_{ij}\) and
\(\{a_i,a_j\}=\{a_i^\dagger,a_j^\dagger\}=0\).
Their operator algebra is \(\Cl(6,\C)\simeq M_8(\C)\).
The number operator and vacuum projector are
\begin{equation}
 N=\sum_i a_i^\dagger a_i,\qquad
 p_0=\prod_i(1-a_i^\dagger a_i).
 \label{eq:clifford}
\end{equation}
The projector is rank one, and the left ideal \(M_8(\C)p_0\) has complex
dimension eight. The eigenvalues of \(N/3\) are \(0,1/3,2/3,1\), with
multiplicities \(1,3,3,1\), providing a useful charge construction
\cite{Furey2015,Furey2018}. The eight-dimensional coefficient space
\(\C\otimes\Oo\) is nonassociative; it is distinct from this 64-dimensional
associative Clifford operator algebra.

For anti-Hermitian octonionic left multiplications \(L_e,L_f\), with
\(L_e^2=L_f^2=-\id\) and \(\{L_e,L_f\}=0\), a consistent adjoint pair is
\begin{equation}
 a^\dagger=\frac{L_e+iL_f}{2},\qquad
 a=\frac{-L_e+iL_f}{2},\qquad \{a,a^\dagger\}=\id .
 \label{eq:ladder}
\end{equation}
The opposite sign for this annihilator gives an anticommutator \(-\id\).
The complex-octonionic elements \((1\pm ie_7)/2\) are idempotents in the
subalgebra generated by \(e_7\), but their left multiplications have rank
four, since \(L_{e_7}^2=-\id_8\) and \(\Tr L_{e_7}=0\).
They are not the primitive rank-one projector \(p_0\).

An explicit one-family chiral benchmark consists of
\begin{equation}
 S_+=(\nu,d_a^c,u_a,e^c),\qquad
 S_-=(\nu^c,d_a,u_a^c,e),\qquad a=1,2,3 ,
 \label{eq:coefficientstates}
\end{equation}
all interpreted as left-handed Weyl coefficients. Assign colour
\(\bm3\) to \(u,d\) and \(\bar{\bm3}\) to \(u^c,d^c\), and define
\begin{equation}
 \begin{split}
 Q&=\tfrac13\diag(N,-N),\\
 T_+&=|\nu\rangle\langle e|+\sum_a|u_a\rangle\langle d_a|,
 \quad T_-=T_+^\dagger,\quad
 T_3=\tfrac12[T_+,T_-],\quad Y=Q-T_3 .
 \end{split}
 \label{eq:weak}
\end{equation}
The weak action connects the two coefficient spaces and is additional
specified structure. It gives the following multiplets:
\begin{center}
\begin{tabular}{@{}cccccc@{}}
\toprule
Fields & \((u,d)\) & \(u^c\) & \(d^c\) & \((\nu,e)\) & \(e^c,\nu^c\)\\
\midrule
\(SU(3)_C,SU(2)_L\) & \((\bm3,\bm2)\) & \((\bar{\bm3},\bm1)\)
& \((\bar{\bm3},\bm1)\) & \((\bm1,\bm2)\) & \((\bm1,\bm1)\)\\
\(Y\) & \(1/6\) & \(-2/3\) & \(1/3\) & \(-1/2\) & \(1,0\)\\
\bottomrule
\end{tabular}
\end{center}
For example, the \(SU(3)^2Y\) and \(SU(2)^2Y\) sums are proportional to
\(2(1/6)-2/3+1/3=0\) and \(3(1/6)-1/2=0\).
The gravitational and cubic hypercharge sums are
\begin{align}
 6(1/6)+3(-2/3)+3(1/3)+2(-1/2)+1&=0,\\
 6(1/6)^3+3(-2/3)^3+3(1/3)^3+2(-1/2)^3+1&=0 .
\end{align}
Colour is vectorlike in its anomaly count, and four weak doublets per
family give an even number for the global \(SU(2)\) anomaly. Three copies
give \eqref{eq:phys}; their multiplicity has been supplied, not derived.
Neither an imaginary internal coefficient nor a decomposition into
self-conjugate spinor combinations establishes a Majorana mass term
\cite{Dreiner2010}. Neutrino mass type is outside the charged-mass ansatz.

\section*{Statements and Declarations}
\paragraph{Funding.} No funding was received for this study.
\paragraph{Competing interests.} The author declares no competing interests.
\paragraph{Author contributions.} T. P. Singh is the sole author and is
responsible for the research and manuscript.
\paragraph{Data and code availability.}
The published numerical inputs, source-precision record, code and computational
certificates are provided in the supplementary material. The empirical
sources and their conventions are identified in \texttt{anc/inputs.json}
and \texttt{anc/provenance.txt}. The standard-library programs
\texttt{reproduce.py}, \texttt{verify\_selectors.py} and
\texttt{verify\_joint\_quarks.py} regenerate the numerical outputs and
support a read-only \texttt{--check} mode. The last program uses directed
decimal arithmetic to verify the complete branch enclosures.
\texttt{verify\_mass\_map.py} uses NumPy and SymPy for additional exact and
numerical checks. The published AHS evolution is an input, not independently
rerun here; no joint likelihood is reconstructed.
\paragraph{Use of artificial intelligence.}
OpenAI ChatGPT/Codex (GPT-6 Astra) and Anthropic Claude Fable 5.2 assisted with mathematical analysis,
source checks, numerical code and manuscript preparation.
AI-assisted adversarial checks do not constitute external peer review.
The human author has reviewed and verified the manuscript,  and takes full intellectual responsibility for its scientific content, attribution and submission.

\end{document}